\documentclass[12pt,a4paper, twoside,reqno]{amsart}
\usepackage{amscd}
\usepackage{amsmath,amssymb,amsthm,mathrsfs}
\usepackage[colorlinks=true,linkcolor=blue,citecolor=blue]{hyperref}
\usepackage[margin=3cm]{geometry}
\usepackage{tikz}
\usetikzlibrary{calc}

\usepackage{parskip}
\usepackage{enumerate}
\newcommand{\dd}[2]{\frac{d #1}{d #2}}

\renewcommand{\H}{\mathcal{H}}
\newcommand{\w}{\widetilde}
\newcommand{\wH}{\widetilde \H }

\newcommand{\pa}{\partial}
\newcommand{\C}{\mathbb C}
\newcommand{\R}{\mathbb R}
\newcommand{\ZZ}{\mathbb Z}

\newtheorem{thm}{Theorem}[section]
\newtheorem*{thm*}{Theorem}
\newtheorem{prop}[thm]{Proposition}
\newtheorem{lem}[thm]{Lemma}

\newtheorem{conj}[thm]{Conjecture}
\theoremstyle{definition}

\theoremstyle{remark}
\newtheorem{remark}[thm]{Remark}
\numberwithin{equation}{section}

\title{Integrability of Limit Shapes of the Inhomogeneous Six Vertex Model}

\author{David Keating}
\address{D.K.: Department of Mathematics, University of California, Berkeley,
CA 94720, USA}
\email{dkeating@berkeley.edu}

\author{Nicolai Reshetikhin}
\address{N.R.: Department of Mathematics, University of California, Berkeley,
CA 94720, USA\\ \& Saint Petersburg University, Russia \\  \& KdV Institute for Mathematics, University of Amsterdam, 1098 XH Amsterdam, The Netherlands .}
\email{reshetik@math.berkeley.edu}

\author{Ananth Sridhar}
\address{A.S.: Google}
\email{ananthsridhar@google.com}

\begin{document}
\begin{abstract}
In this paper we prove that the Euler-Lagrange equations for the limit shape for the inhomogeneous six vertex model on a cylinder have infinitely many conserved quantities.
\end{abstract}
\maketitle

\tableofcontents

\section*{Introduction} 

In the thermodynamic limit, the height function of the homogeneous 6-vertex model develops a limit shape \cite{ZJ,AR,CP,PR}. Moreover,  the limit shape height function can be derived from a variational principle similar to the one proven for bipartite dimer models \cite{CKP}. The Euler-Lagrange equations for this variational principle considered as an evolution equations in Euclidean time admit infinitely many conserved quantities \cite{RS}. For the homogeneous six vertex model on a cylinder the limit shape height function can also be interpreted as a Hamiltonian flow in Euclidean time along the cylinder. The conservation laws are obtained by \cite{RS} in the Hamiltonian framework. They
form a Poisson commutative family of functionals and therefore it is natural to expect that this Hamiltonian system is integrable. 

One can argue that the existence of infinite number of conservation laws for the PDE defining the limit shape is related to the commutativity of transfer matrices and that it is a semiclassical version of this property. The {\it integrable inhomogeneities} that we consider here correspond to shifts of the spectral parameters. They do not change the commutativity 
of transfer-matrices.  This was first observed by Baxter, see \cite{Ba} and was used many times in literature, see for example \cite{FR}.

In this paper we extend the results of \cite{RS} to the case of the six vertex model with integrable inhomogeneities. We prove that Euler-Lagrange equations for critical points of the large deviation functional for this model on a cylinder have infinitely many conserved quantities. As in the homogeneous case we use the Hamiltonian framework. 

Theorem \ref{main1} states the main result on the Poisson commutativity of an infinite family of integrals of motion. The theorem follows from bilinear differential identities
for the semicanonical free energy function. These identities are proven in Theorem \ref{main}.  The proof is based on properties of integral equations involved in the description
of the ground state of the model. These properties are presented in section \ref{ident}. The analysis largely follows \cite{NK,IKR}. The description of the free energy uses standard conjectures about the ground state of the 6-vertex model which were confirmed numerically in numerous cases and in some cases, such as a free fermionic point, are proven.

The organization of the paper is as follows: In section \ref{6vm}, we briefly describe the transfer matrix and partition function of the six vertex model, including a review of the Bethe ansatz for the eigenvectors of the transfer matrix. In section \ref{thlim} we review the structure of the maximal eigenvalue in the thermodynamic  limit.  Integral equations describing the density of the free energy on a torus in the thermodynamic limit are analyzed in section \ref{ident}. In section \ref{sec:4}, we outline the variational principle and the Hamiltonian framework for the limit shape and prove  the main theorem. In the concluding section \ref{concl}, we give concluding remarks and present some open problems.

{\bf Acknowledgements} The work of N.R. was partly 
supported by grants NSF FRG DMS-1664521 and RSF-18-11-00-297. The work of D.K. was 
partly supported by the grant NSF DMS-1902226. N.R. and D.K. are grateful for the hospitality  at ITS ETH where he was visiting 
when the work was completed. N.R. would like to thank A. Borodin, I. Corwin and A. Pronko for helpful
discussions on various aspects of the six vertex model.

\noindent
\section{The six vertex model on a cylinder}\label{6vm}

\subsection{The six vertex model}
Here we recall some basic facts about the six vertex model (see  \cite{Ba}, survey \cite{Res10} and section 1 of \cite{RS} and references therein). Recall that the ice rule implies that the six vertex model can be seen as a set of path which do not cross, but may touch at vertex. Figure \ref{fig:6Vconfig} show the possible states for a single vertex, note that we orient the path up and right. 
\begin{figure}[h]
\begin{tabular}{cccccc}
\begin{tikzpicture}
\draw[->][thick] (-1,0) -- (-.5,0);\draw[thick] (-.5,0) -- (0,0); \draw[->][thick] (0,0) -- (.5,0);\draw[thick] (.5,0) -- (1,0);
\draw[->][thick] (0,-1) -- (0,-.5);\draw[thick] (0,-.5) -- (0,0); \draw[->][thick] (0,0) -- (0,.5);\draw[thick] (0,.5) -- (0,1);
\end{tikzpicture} 
&
\begin{tikzpicture}
\draw[thick] (-1,0) -- (-.5,0);\draw[<-][thick] (-.5,0) -- (0,0); \draw[thick] (0,0) -- (.5,0);\draw[<-][thick] (.5,0) -- (1,0);
\draw[thick] (0,-1) -- (0,-.5);\draw[<-][thick] (0,-.5) -- (0,0); \draw[thick] (0,0) -- (0,.5);\draw[<-][thick] (0,.5) -- (0,1);
\end{tikzpicture} 
&
\begin{tikzpicture}
\draw[->][thick] (-1,0) -- (-.5,0);\draw[thick] (-.5,0) -- (0,0); \draw[->][thick] (0,0) -- (.5,0);\draw[thick] (.5,0) -- (1,0);
\draw[thick] (0,-1) -- (0,-.5);\draw[<-][thick] (0,-.5) -- (0,0); \draw[thick] (0,0) -- (0,.5);\draw[<-][thick] (0,.5) -- (0,1);
\end{tikzpicture} 
&
\begin{tikzpicture}
\draw[thick] (-1,0) -- (-.5,0);\draw[<-][thick] (-.5,0) -- (0,0); \draw[thick] (0,0) -- (.5,0);\draw[<-][thick] (.5,0) -- (1,0);
\draw[->][thick] (0,-1) -- (0,-.5);\draw[thick] (0,-.5) -- (0,0); \draw[->][thick] (0,0) -- (0,.5);\draw[thick] (0,.5) -- (0,1);
\end{tikzpicture} 
&
\begin{tikzpicture}
\draw[->][thick] (-1,0) -- (-.5,0);\draw[thick] (-.5,0) -- (0,0); \draw[thick] (0,0) -- (.5,0);\draw[<-][thick] (.5,0) -- (1,0);
\draw[thick] (0,-1) -- (0,-.5);\draw[<-][thick] (0,-.5) -- (0,0); \draw[->][thick] (0,0) -- (0,.5);\draw[thick] (0,.5) -- (0,1);
\end{tikzpicture} 
&
\begin{tikzpicture}
\draw[thick] (-1,0) -- (-.5,0);\draw[<-][thick] (-.5,0) -- (0,0); \draw[->][thick] (0,0) -- (.5,0);\draw[thick] (.5,0) -- (1,0);
\draw[->][thick] (0,-1) -- (0,-.5);\draw[thick] (0,-.5) -- (0,0); \draw[thick] (0,0) -- (0,.5);\draw[<-][thick] (0,.5) -- (0,1);
\end{tikzpicture} 
\\
\begin{tikzpicture}
\draw[ultra thick] (-1,0) -- (0,0); \draw[ultra thick] (0,0) -- (1,0);
\draw[ultra thick] (0,-1) -- (0,0); \draw[ultra thick] (0,0) -- (0,1);
\end{tikzpicture}
&
\begin{tikzpicture}
\draw (-1,0) -- (0,0); \draw (0,0) -- (1,0);
\draw (0,-1) -- (0,0); \draw (0,0) -- (0,1);
\end{tikzpicture} 
&
\begin{tikzpicture}
\draw[ultra thick] (-1,0) -- (0,0); \draw[ultra thick] (0,0) -- (1,0);
\draw (0,-1) -- (0,0); \draw (0,0) -- (0,1);
\end{tikzpicture}
&
\begin{tikzpicture}
\draw (-1,0) -- (0,0); \draw (0,0) -- (1,0);
\draw[ultra thick] (0,-1) -- (0,0); \draw[ultra thick] (0,0) -- (0,1);
\end{tikzpicture}
&
\begin{tikzpicture}
\draw[ultra thick] (-1,0) -- (0,0); \draw (0,0) -- (1,0);
\draw (0,-1) -- (0,0); \draw[ultra thick] (0,0) -- (0,1);
\end{tikzpicture}
&
\begin{tikzpicture}
\draw (-1,0) -- (0,0); \draw[ultra thick] (0,0) -- (1,0);
\draw[ultra thick] (0,-1) -- (0,0); \draw (0,0) -- (0,1);
\end{tikzpicture}
\\
$w_1$ & $w_2$ & $w_3$ & $w_4$ & $w_5$ & $w_6$
\end{tabular}
\caption{Six vertex configurations.}
\label{fig:6Vconfig}
\end{figure}
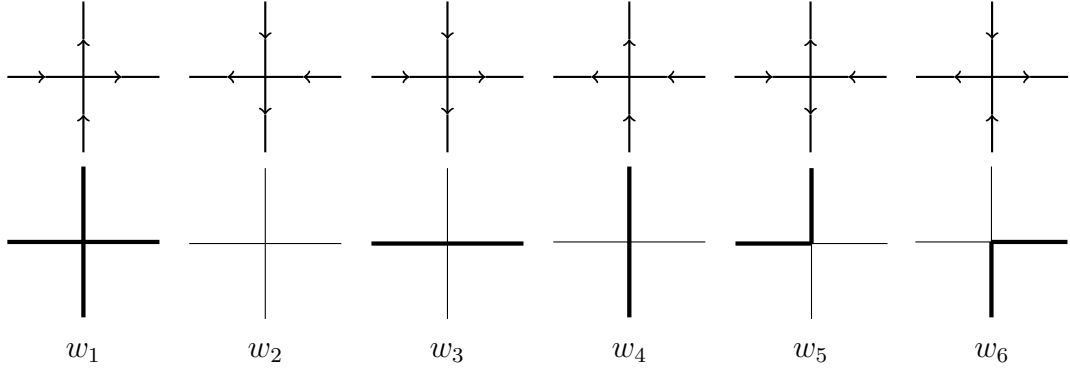

An important parameter of the six vertex model is 
\[
\Delta = \frac{w_1w_2+w_3w_4-w_5w_6}{2\sqrt{w_1w_2w_3w_4}}
\]
Baxter gave the following important parameterizations of the weights for the symmetric six vertex model when $w_1=w_2=a$, $w_3=w_4=b$, and $w_5=w_6=c$.
\begin{enumerate}
\item When $\Delta > 1$:
\begin{enumerate}
\item If $a>b+c$, let $(a,b,c) = ( r\sinh(u+\eta), r\sinh u, r\sinh \eta )$, with $\eta > 0$.
\item If $b>a+c$,  let $(a,b,c) =  ( r\sinh(u-\eta), r\sinh u, r\sinh \eta )$, with $0 < \eta < u$.
\end{enumerate}
Here $\Delta = \cosh \eta$.
\item When $\Delta < -1$: let $(a,b,c) = ( r\sinh(\eta-u), r\sinh u, r\sinh \eta )$ with $0<u<\eta$. Here $\Delta = - \cosh \eta $.
\item When $-1< \Delta < 1$:
\begin{enumerate}
\item If $a>b+c$, let $(a,b,c) =  ( r\sin(u-\gamma),r\sin u, r\sin \gamma )$, with $0 < \gamma < u < \frac{\pi}{2}$. Here $\Delta = \cos \gamma$.
\item If $b>a+c$,  let $(a,b,c) =  ( r\sin(\gamma-u), r\sin u , r\sin \gamma )$, with $0 < u < \gamma < \frac{\pi}{2}$. Here $\Delta = -\cos \gamma$.
\end{enumerate}
\end{enumerate}
where $u$ is called the spectral parameter. 

In the non-symmetric case the weights can be naturally parametrized as 
\begin{equation}
\begin{aligned}
w_1 = ae^{H+V}, & \; w_3 = b e^{H-V}, & \; w_5 = c \lambda \\
w_2 = ae^{-H-V}, & \; w_4 = b e^{-H+V}, & \; w_6 = c \lambda^{-1}
\end{aligned}
\end{equation}
where $H$ and $V$ can be viewed as magnetic (or electric, depending on the interpretation) field. The magnetic fields can be introduced as edge weights in the symmetric model if we assign each occupied horizontal edge a weight of $e^{\frac{H}{2}}$ and each unoccupied $e^{-\frac{H}{2}}$. Similarly, each occupied vertical edge gets a weight of $e^{\frac{V}{2}}$ and each unoccupied $e^{-\frac{V}{2}}$. 

Note that on the cylinder and torus the number of vertices with weight $w_5$ is equal to those with weight $w_6$, so we may set $\lambda=1$ without loss of generality. As we are primarily concerned with these cases, we set $\lambda = 1$ for what follows.

\subsection{The Yang-Baxter equation}
Let $e_1 = \begin{pmatrix} 1\\0 \end{pmatrix}$ and $e_2 = \begin{pmatrix} 0\\1 \end{pmatrix}$ be the standard basis of $\C^2$. To each edge we assign a vector in $\C^2$, with $e_1$ corresponding to an occupied edge, and $e_2$ to an unoccupied edge. Then, in the tensor product basis $e_1\otimes e_1,e_1\otimes e_2,e_2\otimes e_1,e_2\otimes e_2$, we can arrange the six vertex weights into the $4\times 4$ matrix 
\begin{equation}
R(u,H,V) = \begin{pmatrix} ae^{H+V} & 0 & 0 & 0 \\ 0 & b e^{H-V} & c & 0 \\ 0 & c & b e^{-H+V} & 0 \\ 0 & 0 & 0 & ae^{-H-V} \end{pmatrix}
\end{equation}
Here we explicitly write the dependence on the spectral parameter and magnetic fields. Let $R(u) = R(u,0,0)$. We can write $R(u,H,V)$ as
\begin{equation} \label{eq:RuHV}
R(u,H,V) = \left(D^H \otimes D^V\right) R(u) \left( D^H \otimes D^V \right)
\end{equation}
where $D^H = \begin{pmatrix} e^{\frac{H}{2}} & 0 \\ 0 & e^{-\frac{H}{2}} \end{pmatrix}$. The ice rules imply the identity
\begin{equation} \label{eq:Rid1}
\left(D\otimes D\right) R(u,H,V) = R(u,H,V)\left(D\otimes D\right)
\end{equation}
for any diagonal matrix $D$.

In \cite{Ba} Baxter showed that the R-matrix satisfies the Yang-Baxter equation
\begin{equation} \label{eq:YB1}
R_{12}(u) R_{13}(u+v) R_{23}(v) = R_{23}(v) R_{13}(u+v) R_{12}(u)
\end{equation}
in $\C^2\otimes\C^2\otimes\C^2$, where the subscripts indicate which factors of the tensor product the matrix acts on and $R(u) = R(u,0,0)$.  This along with equations (\ref{eq:RuHV}) and (\ref{eq:Rid1}) imply that
\begin{equation} \label{eq:YB2}
R_{12}(u) R_{13}(u+v,H,0) R_{23}(v,H,0) = R_{23}(v,H,0) R_{13}(u+v,H,0) R_{12}(u)
\end{equation}

\subsection{The transfer matrix and partition function} 
Using the $R$-matrix above, we can construct the partition function for the six vertex model on a cylinder.

Consider the inhomogeneous 6-vertex model on a cylinder with inhomogeneity parameters $v_k, k=1, \cdots, N$ corresponding to vertical lines of the lattice. That is, for a single row the spectral parameter at site $k$ is given by $u-v_k$. Construct that quantum monodromy matrix $T_a(u,\{v_k\},H,0):\C^2\otimes (\C^2)^{\otimes N} \to \C^2\otimes (\C^2)^{\otimes N}$ by 
\begin{equation}\label{tr-matrix}
\begin{aligned}
T_a(u,\{v_k\}, H,0) &= D^{2H}_aR_{1a}(u - v_1)D^{2H}_a R_{2a}(u-v_2) \cdots D^{2H}_aR_{Na}(u-v_N)  \\ &=\begin{pmatrix}A(u)&B(u)\\C(u)&D(u)\end{pmatrix}.
\end{aligned}
\end{equation}
Here, the first factor is enumerated by $a$, others by $1,\cdots, N$. The matrix elements of $T_a(u, H,0)$ can be thought of as weight of the configuration on a single row with given boundary conditions. Taking the trace over the first factor and adding a vertical magnetic field we have the row-to-row transfer matrix $t(u,\{v_k\}, H, V): (\C^2)^{\otimes N}\to (\C^2)^{\otimes N}$
\begin{align} \label{eq:transfermatrix}
t(u,\{v_k\}, H, V) = (D^{2V}_1\cdots D^{2V}_N\; )\text{Tr}_a \; T_a(u,\{v_k\}, H, 0).
\end{align}
The elements of this can be seen as the weight of a configuration of a single row on the cylinder with specified boundary conditions and magnetic fields $H$ and $V$. Finally, we can construct the partition function on a cylinder with $M$ row with free boundary conditions and inhomogeneity parameters $u_i, i=1,\cdots, M$ corresponding to the horizontal lines of the lattice
\begin{align*}
Z_{MN}^{cyl}(\{u_i\},\{v_k\},H,V) = t(u_1,\{v_k\},H,V) t(u_2,\{v_k\},H,V) \cdots t(u_M,\{v_k\},H,V).
\end{align*}
The partition function with boundary states $\xi_1$ and $\xi_2$ is
{\small
\begin{align*}
Z_{{MN}, \xi_1, \xi_2 }^{cyl}(\{u_i\},\{v_k\},H,V) = (\xi_1, \; t(u_1,\{v_k\},H,V) t(u_2,\{v_k\},H,V) \cdots t(u_M,\{v_k\},H,V) \xi_2 ).
\end{align*}
}
where $(x,y)$ is the natural scalar product on $(\C^2)^{\otimes N}$. Note that equations (\ref{eq:Rid1}) and (\ref{eq:transfermatrix}) imply that in the case of the cylinder the vertical magnetic field only contributes an overall factor of $e^{M(N-2n)V}$ to the partition function, where $n$ is number of vertical edges occupied by paths in each row. Because of this we will focus on the partition function of a cylinder with $V=0$. 

If we then take the trace over $ (\C^2)^{\otimes N} $, we get the partition function for the $M\times N$ torus
\begin{align*}
Z_{{MN}}^{torus}(\{u_i\},\{v_k\},H,V) = Tr\left( Z_{MN}^{cyl}(\{u_i\},\{v_k\},H,V) \right).
\end{align*}
Note that 
\begin{align*}
Z_{{MN}}^{torus}(\{u_i\},\{v_k\},H,V) = \sum_{n=0}^N e^{M(N-2n)V} Z_{{MN}}^{torus, n}(\{u_i\},\{v_k\},H).
\end{align*}
According to the terminology in statistical mechanics, it is natural to call $Z_{{MN}}^{torus}(u,H,V)$ the grand canonical partition function, where we sum over all possible topological configurations of paths on the torus. The partition function $Z_{{MN}}^{torus, n}(u,H)$ in this sense should be called the semigrand canonical partition function, where the number of paths crossing a horizontal cycle is fixed by $n$.

\subsection{The spectrum of transfer-matrices for finite $N$}
Here we assume $ \Delta=\frac{(a^2+b^2-c^2)}{2ab} < -1 $. In this region Baxter's parametrization is 
\begin{align*}
a = \sinh(\eta-u), \; \; b = \sinh u, \;\ \; c = \sinh \eta.
\end{align*}
In this parametrization $\Delta = -\cosh \eta$ and $0<u<\eta$. Later we will comment on other values of $\Delta$.

The following construction is known as an algebraic Bethe ansatz\footnote{ The idea of using a superposition  of plane waves as an eigenfunction for the Heisenberg spin Hamiltonian (which is the the logarithmic derivative of the transfer-matrix of the homogeneous 6-vertex model at $u=0$) goes back to H. Bethe \cite{B}. It was first applied to the 6-vertex model by E. Lieb in \cite{L} for zero magnetic fields. Shorty after C.P. Yang \cite{CPY} applied it to the asymmetric 6-vertex model (with magnetic fields). The algebraic form we use is due to L. Faddeev and L. Takhtajan \cite{FT}.}. It states that if $\{\alpha_j\}$ satisfy the Bethe equations
\begin{equation} \label{eq:BE}
\begin{aligned}
\prod_{k = 1}^N \frac{\sinh( \frac{\eta}{2}+ i \alpha_j - v_k )}{\sinh( \frac{\eta}{2} - i \alpha_j + v_k  ) } = e^{2 H N } \prod_{m=1,m \neq j }^n \frac{\sinh( i (\alpha_j -  \alpha_m) + \eta )}{\sinh( i (\alpha_j -  \alpha_m) - \eta ) }
\end{aligned}
\end{equation}
then the vector
\begin{equation}
\begin{aligned}
B(\alpha_1)\dots B(\alpha_n) \;  \Omega,
\end{aligned} \begin{aligned} \hspace{20pt} \Omega=\begin{pmatrix} 0 \\ 1 \end{pmatrix}\otimes \cdots
\otimes \begin{pmatrix} 0 \\ 1 \end{pmatrix}
\end{aligned}
\end{equation}
is an eigenvector of the transfer-matrix (\ref{eq:transfermatrix}) with the eigenvalue
\begin{equation}\label{eq:Lambda}
\begin{aligned}
\Lambda(u,\{v_k\},H,0) =&  \;e^{N H} \prod_{k = 1}^N \sinh( \eta - u + v_k )  \prod_{j=1}^n \frac{\sinh(  \frac{\eta}{2} + u - i \alpha_j )}{\sinh(  \frac{\eta}{2} - u + i \alpha_j  ) }  \\&+  e^{-N H} \prod_{k = 1}^N \sinh( u - v_k)  \prod_{j=1}^n \frac{\sinh(\frac{3 \eta}{2} - u + i \alpha_j )}{\sinh( u - \frac{\eta}{2} - i \alpha_j  ) }.
\end{aligned}
\end{equation}
This is the description of the spectrum of the transfer-matrix for the model in the horizontal magnetic field $H$ with inhomogeneities $ \{ v_k \} $ for $\Delta<-1$.  

Define functions $ p $ and $ \Theta $ defined by,
\begin{equation} \label{def:pTheta}
\begin{aligned} 
e^{i p(\alpha) }  &= \frac{e ^{\eta + 2 i \alpha} - 1 }{e^{\eta} - e^{2 i \alpha} } \\ &= \frac{\sinh( \frac{\eta}{2} + i \alpha )}{\sinh( \frac{\eta}{2} - i \alpha ) }, \\
e^{i \Theta (\alpha- \beta) } &= \frac{1 + e^{i p(\alpha) + i p(\beta) } - 2 \Delta e ^{i p(\alpha) } }{1 + e^{i p(\alpha) + i p(\beta) } - 2 \Delta e ^{i p(\beta)}} \\ &= - \frac{\sinh(i \alpha - i \beta + \eta) }{ \sinh( i \alpha - i \beta - \eta) }.
\end{aligned}
\end{equation}
The following identities can be checked directly:
\begin{equation} \label{eq:kernsym}
\begin{aligned}
p(\alpha^*) &= p(\alpha)^*\\
\Theta(\alpha^*) &= \Theta(\alpha)^*
\end{aligned} \; \; \; \; \; \;
\begin{aligned}
p(-\alpha) &= -p(\alpha) \\
 \Theta(-\alpha) &= - \Theta(\alpha)
\end{aligned}
\end{equation}
Here $a^*$ is the complex conjugate of $a$.
In terms of these functions Bethe equations can be conveniently written as
\begin{align*}
\prod_{k=1}^N e^{i p (\alpha_j + i v_k ) } = e^{2 H N } \prod_{m=1,m \neq j}^n \left( - e^{i \Theta( \alpha_j-\alpha_m ) } \right).
\end{align*}

\begin{remark} The spectrum of the transfer-matix for other values of $\Delta$ can be obtained  similarly, using the analytical continuation.
In terms of functions $a(u),b(u),c(u)$ parametrizing weights of the model, the eigenvalues of the transfer matrix for all values of $\Delta$ can be written as 
\[
\Lambda(u,\{v_k\},H,0) = e^{NH} \prod_{k=1}^N a(u-v_k)  \prod_{i=1}^n \frac{a(\lambda_i - u)}{b(\lambda_i-u)} + e^{-NH} \prod_{k=1}^N b(u-v_k) \prod_{i=1}^n \frac{a(u-\lambda_i)}{b(u-\lambda_i)}
\]
encompassing any choice of parametrization, where the $\lambda_j$ solve the Bethe equations
\[
\prod_{k=1}^N \frac{b(\lambda_j-v_k)}{a(\lambda_j-v_k)} = e^{2NH} (-1)^n\prod_{i=1,i\ne j}^n \frac{a(\lambda_i-\lambda_j)}{a(\lambda_j-\lambda_i)}.
\]
Equations (\ref{eq:Lambda}) and (\ref{eq:BE}) follow from choosing the parametrization for $\Delta < -1$ and letting $\lambda_j = i\alpha_j + \frac{\eta}{2}$. For details, see, for example, \cite{KBI}.
\end{remark}

Introduce notation 
\begin{equation}\label{def:psipm}
\begin{aligned}
e^{\psi_+ (\alpha + i u) } &= \frac{e^{\eta+2 u } - e^{2 i \alpha} }{ e^{ \eta - 2 i \alpha } - e^{2 u} }= \frac{\sinh( \frac{\eta}{2} + u - i \alpha )}{\sinh( \frac{\eta}{2} - u + i \alpha ) } \\
e^{\psi_- (\alpha + i u) } &= \frac{  e^{2 \eta + 2 i \alpha} - e^{2 u - \eta}}{e^{2 u} - e^{\eta + 2 i \alpha} }= \frac{\sinh(  \frac{3 \eta}{2} - u + i \alpha ) }{\sinh(u - \frac{\eta}{2} - i \alpha ) }.
\end{aligned}
\end{equation}

Then the formula for eigenvalues of in terms of solutions to Bethe equations can be written as
\begin{equation} \label{eq:eigenvalue}
\begin{aligned}
\Lambda(u,\{v_k\},H,0) =& \; e^{N H } \prod_{k = 1}^N \sinh( \eta - u + v_k)   \prod_{j=1}^n e^{\psi_+( \alpha_j + i u ) } \\ & + e^{-N H } \prod_{k = 1}^N \sinh( u - v_k) \prod_{j=1}^n e^{\psi_-(\alpha_j + i u) }.
\end{aligned}
\end{equation}

\section{The thermodynamic limit}\label{thlim}
In this section we will describe the asymptotic of the partition function when $N,M\to\infty$ using the analysis of the Bethe equations in the limit $N\to \infty$. We will first describe the case of homogeneous weights.

\subsection{The ground state for finite $N$}
Recall the Bethe equations are
\begin{align*}
\prod_{k=1}^N e^{i p (\alpha_j + i v_k) } = e^{2 H N} (-1)^{n-1} \prod_{m=1,m \neq j }^n e^{i \Theta( \alpha_j - \alpha_m ) }.
\end{align*}
Choosing branches of logarithms in the definition of $p(\alpha)$ and $\Theta(\alpha-\beta)$ we can write them as
\begin{align*}
\sum_{k=1}^N \frac{1}{N} p(\alpha_j + i v_k) = -2 i H + \frac{2 \pi I_j}{N} + \sum_{m=1,m \neq j }^n \frac{1}{N}\Theta(\alpha_j - \alpha_m )
\end{align*}
where $I_j$ are integers for odd $n$ and half-integers for even $n$.

The first basic conjecture in the Bethe ansatz description of the largest eigenvalue of the transfer
matrix, in the subspace with fixed $n$, states that
\begin{conj} \label{conj:Betheroots1}
When the inhomogeneities and $H$ are sufficiently small, or when the system is homogeneous, the maximal eigenvalue corresponds to solutions of Bethe equations
with
\begin{align*}
 I_j = \frac{n+ 1 - 2 j }{2}, \hspace{15pt} j = 1, \cdots, n
\end{align*}
\end{conj}
This conjecture has a long history. Perhaps, the first detailed study of this conjecture, applied to the Heisenberg XXZ spin chain was done in \cite{YY} where the authors also give an
account of prior results. It was used to characterize the ground state when $N\to\infty$ and to compute the free energy in this limit in \cite{CPY}\cite{SYY}.  An account of this and other works on the 6-vertex model can be found in a survey \cite{LW} and in \cite{Ba}.  Among more recent results are:  a detailed exposition of results from \cite{CPY}\cite{SYY}
was done in \cite{N}, for a rather detailed study of analytic properties of solutions to Bethe equations for SOS model with twisted boundary conditions see \cite{BM}, an explicit description of the ground state energy and the free energy for the 5-vertex model was found in \cite{GKW}.

\subsection{The ground state in the limit $N\to \infty$, homogeneous case} 
The following conjecture describes the behavior of the solution to the Bethe equations corresponding
to the maximal eigenvalue of the transfer-matrix in the  limit when $N\to \infty$ and the ration $n/N$ is fixed.
\begin{conj} \label{conj:Betheroots2}
Fix the ratio $\frac{n}{N}$. As  $N \rightarrow \infty $, the roots of Bethe equations  corresponding to $I_j$ from conjecture \ref{conj:Betheroots1} become distributed along a contour $ C $ which is described below.
\end{conj}

To describe $C$, let us we introduce
\begin{equation}
\begin{aligned}
2 \pi t_j = 2 \pi \frac{I_j}{N}
\end{aligned} \hspace{20pt}
\begin{aligned}
 -\frac{q}{2} \leq t_j \leq \frac{q}{2}
\end{aligned}
\end{equation}
where $q=n/N$. We will write $\alpha(t_j)$ for $\alpha_j$.

As $ N \rightarrow \infty $ with $ q $ fixed, according to the conjecture \ref{conj:Betheroots2}, the roots $\alpha(t_j)$ of Bethe equations form a complex-valued function $ \alpha(t ) $, $t\in [-\frac{q}{2},\frac{q}{2}]$. The Bethe equations become the non-linear integral equation for $\alpha(t)$
\begin{align}
2 \pi t = p( \alpha(t) ) + 2 H i - \int_{- q / 2}^{q /2 } \Theta( \alpha( t)  - \alpha(s) ) \; ds.
\end{align}
The image of the function $\alpha$ is the contour $C$. This contour connects the endpoints $A=\alpha(-q/2)$ and $B=\alpha(q/2)$. Note that $A=A(q,H)$ depends on both the density and the magnetic field (and similarly for $B$).

Conversely, denote by $ t(\alpha): C \rightarrow  [-\frac{q}{2},\frac{q}{2}] $ the inverse of $ \alpha(t) $. The function $ t $ can be analytically continued off the contour, yielding a complex valued function $ t(\alpha)$ which analytic away from singularities  and branch cuts. Define $\rho$ to be the density of roots of Bethe equations along $ C $. We have
\begin{align} \label{eq:partialt}
\rho(z) = \frac{\partial t(z)}{\partial z } \Big|_{z \in C}.
\end{align} Note that $\text{Im}\big(\rho(z) \; dz|_C \big) = 0$ and
\begin{align} \label{eq:norm}
\int_A^B \; \rho(\alpha) \; d\alpha = q.
\end{align}

From the Bethe equations we obtain the equation for $t(\alpha)$ given by
\begin{align} \label{eq:star}
2 \pi t(\alpha) = p(\alpha) + 2 i H - \int_{A}^{B} \Theta(\alpha- \beta) t'(\beta) d \beta.
\end{align}
Here $\alpha\in C$ and the integral is taken along $C$.

Note that $|\text{Re}(t(\alpha))|\leq  \frac{q}{2}$ and $t(B) =  \frac{q}{2}  , t(A) =  -\frac{q}{2}$. 

Suppose $\alpha$ is on the contour so that $t(\alpha)$ is real. Conjugating equation (\ref{eq:star}), we have 
\begin{align*}
2 \pi t(\alpha)^* &= p(\alpha)^* - 2 i H + \int_{C} \Theta(\alpha - \gamma)^* t'(\gamma)^* d \gamma^* \\ &= p(\alpha^*) - 2 i H +\int_{C} \Theta(\alpha^* - \gamma^*) t'(\gamma^*) d \gamma^*   \\ &=  -p(-\alpha^*) - 2 i H + \int_{C^*} \Theta(\alpha^* - \gamma) t'(\gamma) d \gamma 
\\ &=  -p(-\alpha^*) - 2 i H - \int_{C} \Theta(\alpha^* + \gamma) t'(-\gamma) d \gamma \\
 &=  -p(-\alpha^*) - 2 i H + \int_{C} \Theta(-\alpha^* - \gamma) t'(-\gamma) d \gamma \\&= - 2 \pi t(-\alpha^* )
\end{align*}
where we used the the fact that we extended $t$ analytically so that $t'(\gamma)^*=t'(\gamma^*)$. Thus we see that $-\alpha^*$ lies on the contour and $t(\alpha) = - t(-\alpha^* )  $. It follows that the roots are distributed symmetrically with respect to reflecting across the imaginary axis. In particular, $B=-A^*$.

Differentiating equation (\ref{eq:star}) we obtain an integral equation for $\rho(\alpha)$ given by
\begin{equation} \label{eq:inteqrho}
2 \pi \rho(\alpha) = p'(\alpha) - \int_{C} K(\alpha - \beta) \rho(\beta) d \beta
\end{equation}
where $ K = \Theta' $.  Note that the kernel in the integral equation (\ref{eq:inteqrho}) is a meromorphic function of $\alpha$. This means that the contour can be deformed as long as it does not cross a pole of $K(\alpha-\beta)$. The condition (\ref{eq:partialt}) selects a representative of the equivalence class of continuous deformations of the contour, such that $\rho(\alpha) \; d\alpha$ is a positive density.

Given the solution to (\ref{eq:inteqrho}), the equation (\ref{eq:star}) defines $t(\alpha)$.

\subsection{Maximal eigenvalue in the thermodynamic limit}
Recall the form of the eigenvalues (\ref{eq:eigenvalue}) of the homogeneous transfer matrix 
\[
\begin{aligned}
\Lambda(u,H,0) =\Lambda_+(u,H,0) + \Lambda_-(u,H,0)
\end{aligned} 
\]
where
\[
\begin{aligned}
\Lambda_+(u,H,0) = &  e^{N H } \sinh( \eta - u )^N   \prod_{j=1}^n e^{\psi_+( \alpha_j + i u ) }
\end{aligned}
\]
and
\[
\begin{aligned}
\Lambda_-(u,H,0) = & e^{-N H }\sinh( u )^N \prod_{j=1}^n e^{\psi_-(\alpha_j + i u) }.
\end{aligned}
\]
For generic values of $H$, only one of $\Lambda_\pm$ contributes in the thermodynamic limit.

Define
\begin{align*}
\psi_{u}^{(\pm)}(\alpha) &= \psi_{\pm}(\alpha+ i u )
\end{align*}
where $\psi_{\pm}$ are defined by equation (\ref{def:psipm}). These functions satisfy
\begin{align*}
\psi_{u}^{(\pm)}(\alpha)^* &= \psi_{u}^{(\pm)}(-\alpha^*).
\end{align*}
Define as well 
\begin{equation}
\begin{aligned}
l_{+} = \ln \sinh( \eta - u ),
\end{aligned} \; \; \; \;
\begin{aligned}
l_- = \ln \sinh u,
\end{aligned}
\end{equation}
where we recall that in our parametrization $0<u<\eta$.

Fix $H$ and the ratio $ q = \frac{n}{N}$. As $N \rightarrow \infty $, one of the two terms, $\Lambda_+$ or $\Lambda_-$ dominates and for
the largest eigenvalue of the transfer matrix is
\begin{align*}
\Lambda_{\text{max} }(u, q, H)  = e^{ N \mathcal{H}_u( q, H ) }(1+o(1)),
\end{align*}
where the function $\mathcal{H}_u( q, H )$ is the semigrand canonical free energy of the six vertex model
\begin{align} \label{eq:semigrand}
\mathcal{H}_u(q, H ) = \max_{\pm} \mathcal{H}_u^{\pm}(q, H )
\end{align}
where
\begin{align} \label{eq:HqH}
\mathcal{H}_u^{\pm}(q, H ) =  \pm H + l_{\pm} + \int_{C} \psi_u^{\pm}(\alpha) \rho(\alpha) d\alpha
\end{align}
is the limit of  $\frac{1}{N}\ln(\Lambda_\pm(u,q,H))$ when $q$ is fixed, and $ \rho(\alpha) $ is the density of Bethe roots along the contour $ C $ where they concentrate in the limit $ N \rightarrow \infty $.

\subsection{The free energy of the six vertex model on a torus}
Recall that grand canonical partition function of the six vertex model on a torus can be written as
\begin{equation}\label{ZNM}
Z_{{MN}}^{torus}(u,H,V) = \sum_{n=0}^N e^{M(N-2n)V} Z_{{MN}}^{torus, n}(u,H)
\end{equation}
where $Z_{{MN}}^{torus, n}(u,H)$ is the semigrand canonical partition function. The semigrand canonical partition function itself may be written as
\[
Z_{{MN}}^{torus, n}(u,H) = \sum_{\{\alpha_i\}} (\Lambda_{\{\alpha_i\}}(u,H))^M
\]
where the sum is over all collections $\{\alpha_1,\ldots,\alpha_n\}$ that solve the Bethe equations, and $ \Lambda_{\{\alpha_i\}}$ is the eigenvalue of the transfer matrix $t(u,H,0)$ corresponding to these Bethe roots.

Thus, when  $N,M\to \infty$ with $q=\frac{n}{N}$ fixed, and $M>>N$, one expects the asymptotic of $Z_{{MN}}^{torus, n}(u,H)$ is determined by the contribution from the maximum eigenvalue. From this we have
\begin{align}
Z_{{MN}}^{torus, n}(u,H) = e^{NM \mathcal{H}_u(q, H ) }(1+o(1))
\end{align}
with $\mathscr{H}_u(q, H ) $ the semigrand free energy defined in (\ref{eq:semigrand}). 
One can argue that this asymptotic is uniform in $\frac{N}{M}$ and is given by the same 
formula when $N,M\to \infty$ and the ratio $\frac{N}{M}$ is finite. Combining this asymptotic with (\ref{ZNM})
we obtain the following asymptotic of the grant canonical partition function 
as $N,M\to \infty$
\begin{align}
Z_{{MN}}^{torus}(u,H,V) = e^{NMV} e^{NM f_u(H,V)}(1+o(1)).
\end{align}
Here the grand canonical free energy is the Legendre transform of the semigrand canonical free energy with the vertical magnetic field $V$ conjugate to $q$:
\begin{align}
f_u(H,V) = \max_q (\mathscr{H}_u(q, H ) - 2qV).
\end{align}

\section{Analysis of density integrals and integral equations}\label{ident}
In this section we will study integral equations that appear from Bethe equations in the thermodynamic limit and prove the properties that we use in the proof of the Poisson commutativity of Hamiltonians. In table \ref{tab:def}, below, we summarize multiple notations that we use.

\begin{table}[h!]
\begin{center}
\def\arraystretch{1.5}
    \begin{tabular}{||c|c|c||}
        \hline 
        Function & Definition & Location \\ [0.5ex]
        \hline \hline 
        \small{$\Theta(\alpha - \beta)$} &\small{ $ e^{i\Theta(\alpha - \beta)} = - \frac{\sinh(i(\alpha-\beta) + \eta)}{\sinh(i(\alpha-\beta) - \eta)}$ }&\small{ Def. (\ref{def:pTheta}) }\\ [1ex]
        \hline
       \small{ $K(\alpha - \beta)$ }&\small{ $K(\alpha - \beta) = \frac{\partial}{\partial \alpha} \Theta(\alpha-\beta)$ }&\small{ Eqn. (\ref{eq:inteqrho}) }\\[1ex]
        \hline
       \small{ $p(\alpha)$ }&\small{ $ e^{ip(\alpha)} = \frac{\sinh( \frac{\eta}{2} + i \alpha )}{\sinh( \frac{\eta}{2} - i \alpha ) } $ }&\small{ Def. (\ref{def:pTheta})} \\[1ex]
        \hline
        \small{$\psi_\pm(\alpha +i u)$ }&\small{ $ e^{\psi_\pm (\alpha + iu)} =  \frac{\sinh(u - i\alpha - \frac{\eta}{2} \pm \eta)}{\sinh( \frac{\eta}{2}-u-i\alpha )} $ }&\small{  Def. (\ref{def:psipm}) }\\[1ex]
        \hline
       \small{ $\H_u^\pm(q,H)$} & \small{ $\H_u^\pm(q,H)= \pm H + \ln\; \sinh(  \frac{\eta}{2} \pm (\frac{\eta}{2} - u)) + \int_C\psi_\pm(\alpha +i u) \rho(\alpha)d\alpha$ }&\small{  Eqn. (\ref{eq:HqH}) }\\[1ex]
        \hline
       \small{ $F(\alpha, \gamma)$ }&\small{ $\left(I + \frac{1}{2\pi} K\right)*F = \frac{1}{2\pi} \Theta$ }&\small{ Eqn. (\ref{eq:bigf}) }\\[1ex]
        \hline
        \small{$R(\alpha, \gamma)$ }&\small{$R(\alpha, \gamma) = \frac{\partial}{\partial \gamma} F(\alpha, \gamma)$ }&\small{ Eqn. (\ref{eq:bigr})}\\[1ex]
        \hline
        \small{$D_+(\alpha)$ }&\small{ $\left(I + \frac{1}{2\pi} K\right)*D_+(\alpha) = \frac{1}{2\pi} \left(\Theta(\alpha-B)+\Theta(\alpha-A)\right)$ }&\small{ Eqn. (\ref{eq:dpm})}\\[1ex]
        \hline
        \small{$D_-(\alpha)$ }&\small{ $\left(I + \frac{1}{2\pi} K\right)*D_-(\alpha) = \frac{1}{2\pi} $ }&\small{ Eqn. (\ref{eq:dpm}) }\\[1ex]
        \hline
        \small{$\xi$ }&\small{ $\xi = 2\pi i \left(f'(B) + \int_C f'(\alpha)R(\alpha,B)d\alpha \right) $ }&\small{ Eqn. (\ref{eq:xi}) }\\[1ex]
        \hline
        \small{$\w \xi$ }&\small{ $\w \xi = 2\pi i \left(f'(A) + \int_C f'(\alpha)R(\alpha,A)d\alpha \right) $ }&\small{ Eqn. (\ref{eq:xi}) }\\[1ex]
        \hline
    \end{tabular}
    \caption{The definitions of various important functions used in the paper. }\label{tab:def}
\end{center}
\end{table}

\subsection{Integrals kernels and integral equations}

The integral equation (\ref{eq:inteqrho}) for $ \rho $ can be written as 
\begin{align*}
\rho + \left( \frac{1}{2 \pi} K  * \rho\right) = \frac{1}{2 \pi} p'
\end{align*}
or
 \begin{align*}
 \left( I + \frac{1}{2 \pi} K \right) * \rho = \frac{1}{2 \pi} p'
 \end{align*}
 where $ * $ is convolution, $I$ is the convolution identity, and we recall $ K(\alpha) = \Theta'(\alpha) $.

Define $ F $ by
\begin{align} \label{eq:bigf}
F(\alpha, \gamma) + \frac{1}{2\pi} \int_{C} K(\alpha - \beta) F(\beta, \gamma) \; d \beta = \frac{1}{ 2 \pi } \Theta(\alpha - \gamma).
\end{align}
Let $ R(\alpha, \gamma) = \frac{\partial}{\partial \gamma} F(\alpha, \gamma)  $. It satisfies
\begin{align} \label{eq:bigr}
R(\alpha, \gamma) + \frac{1}{2\pi} \int_{C} K(\alpha - \beta) R(\beta, \gamma) d \beta = - \frac{1}{2 \pi} K(\alpha - \gamma).
\end{align}

In other words,
\begin{align*}
& R + \frac{1}{2 \pi } K * R = - \frac{1}{2 \pi } K  \\
\Rightarrow \; \; & (I + \frac{1}{ 2\pi} K) * R =  - \frac{1}{2 \pi } K \\
\Rightarrow \; \; & (I + \frac{1}{ 2\pi} K) * R + (I + \frac{1}{2 \pi } K) = I \\
\Rightarrow \; \; & (I + \frac{1}{ 2\pi} K) * (I+R) = I
\end{align*} or as operators
\[
I + R = \left( I + \frac{1}{2 \pi} K \right)^{-1}.
\]

We see that $ F $ and $ R $ satisfy
\begin{align}
\left( I + \frac{1}{2 \pi } K \right) * F &= \frac{1}{2 \pi} \Theta \label{eq:bigfresolv}\\
 \left( I + \frac{1}{2 \pi} K \right) * \left(I + R \right) &= I \label{eq:bigrresolv}
\end{align}

\begin{lem}
The functions satisfy the following symmetries:
\begin{enumerate}[a)]
\item $ K(\alpha^*) = K(\alpha) $
\item $ F(-\alpha^*, - \beta^* ) = - F(\alpha, \beta)^*  $
\item $ \rho(\alpha)^* = \rho(- \alpha^* ) $
\item $ R(\alpha, \beta)^* = R(-\alpha^*, -\beta^* ) $
\end{enumerate}
\end{lem}

\begin{lem}\label{eq:pfpa}
\begin{align*}
\frac{\partial F(\alpha, \gamma)}{\partial \alpha}  = - R(\alpha, \gamma) -  R(\alpha, B) F(B, \gamma) + R (\alpha, A) F(A, \gamma)
\end{align*}
\end{lem}
\begin{proof}
Starting with the definition (\ref{eq:bigf}) of $ F $ and taking the $\alpha$ derivative, we have
\begin{align*}
\frac{1}{2 \pi} K(\alpha - \gamma) =&\frac{\partial F}{\partial \alpha} + \frac{1}{2 \pi} \int_C \frac{\partial}{\partial \alpha} K(\alpha - \beta) F(\beta, \gamma) d \beta  \\
=& \frac{\partial F}{\partial \alpha} - \frac{1}{2 \pi} \int_C \frac{\partial}{\partial \beta} K(\alpha - \beta) F(\beta, \gamma) d \beta \\
=& \frac{\partial F}{\partial \alpha} -\frac{1}{2 \pi} \left( K(\alpha - B) F(B, \gamma) -K(\alpha - A) F(A,\gamma) \right) + \frac{1}{2 \pi} \int_C K(\alpha - \beta) \frac{\partial F(\beta, \gamma)}{\partial \beta} d \beta.
\end{align*}
Equivalently
\begin{align*}
\frac{1}{2 \pi} \Big( K(\alpha - \gamma)  + K(\alpha - B) F(B, \gamma) -K(\alpha - A) F(A,\gamma) \Big) = \left( I + \frac{1}{2 \pi} K \right) * \frac{\partial F}{\partial \alpha} (\alpha, \gamma).
\end{align*}
Multiplying by $ (I + R) $ and using
\begin{align}\label{eq:rk}
\left( (I + R) * \frac{1}{2\pi } K \right)(\alpha, \beta) &= -R(\alpha, \beta)
\end{align}
we have the lemma.
\end{proof}

\subsection{Analysis of density integrals}

In the next section we will study integrals of the form
\begin{align*}
g(q, H) = \int_{C} f(\alpha) \; \rho(\alpha)  \;  d \alpha
\end{align*}
for arbitrary functions $ f $, where $\rho$ and $C$ are determined by Bethe integral equation (\ref{eq:inteqrho}), as well as conditions (\ref{eq:partialt}) and (\ref{eq:norm}). We will also call these density integrals. Integrals of this type describe largest eigenvalue in the thermodynamic limit (see equations (\ref{eq:semigrand}) and (\ref{eq:HqH})). This section is a survey of \cite{IKR} and \cite{NK}.

\begin{lem}
$ g(q, H ) $ is real if $ f(\alpha)^* = f(-\alpha^*) $.
\end{lem}
\begin{proof}
Since $ -C^* =  C $
\begin{align*}
g^* &= \int_{C^*} f(\alpha)^* \; \rho(\alpha)^* \; d \alpha^* \\
&= \int_{-C^*}f(\gamma) \; \rho(\gamma) \; d \gamma \\
&= \int_{C} f(\gamma) \; \rho(\gamma) \; d \gamma = g
\end{align*}
\end{proof}
\noindent
We will assume $f(\alpha)^*=f(-\alpha^*)$ since the functions $\psi_\pm(\alpha)$ satisfy this assumption.

\subsection{First derivatives of $ g (q, H ) $ } Here we will compute first derivatives
of $g$ in $q$ and $H$. Recall that the endpoints of the contour also depend on $q$ and $H$. Let $dg (q,H) = \frac{\partial g }{\partial q} dq + \frac{\partial g}{\partial H} dH$ be the total derivative with respect to $q$ and $H$. 

For the total derivative of $g(q, H)$ we have
\begin{equation}
\begin{aligned}
d g(q, H) &= d \int_{A}^{B} f(\alpha) \rho(\alpha) d \alpha \\
&= d B \; f(B) \; \rho(B) - d A \; f(A)  \; \rho(A) + \int_{A}^B f(\alpha) \; d \rho(\alpha)  \; d \alpha
\end{aligned}
\end{equation}
Recall from equation (\ref{eq:partialt}) that $\rho(\alpha)$ is the partial derivative of $t(\alpha)$ along the contour $C$. We can use this in the above and integrate by parts to obtain
\begin{equation} \label{eq:denInt}
    \begin{aligned}
    d g(q, H) &=  d B \; f(B) \; \rho(B) - d A \; f(A)  \; \rho(A) + \int_{A}^B f(\alpha) \; \frac{\partial dt(\alpha)}{\partial \alpha}  \; d \alpha \\
    &=  d B \; f(B) \; \rho(B) - d A \; f(A)  \; \rho(A) + f(B) dt(B) - f(A) dt(A) \\
    & \hspace{1cm} - \int_A^B \frac{\partial f(\alpha)}{\partial \alpha} \; dt(\alpha) \; d\alpha 
    \end{aligned}
\end{equation}

We know that at the endpoints we have
\[
\begin{aligned}
t(B) = \frac{1}{2}q \\
t(A) = -\frac{1}{2}q
\end{aligned}
\] 
(see equation (\ref{eq:star})). Varying these equations in $ A$ and $B$ we obtain
\begin{equation}\label{deltatAB}
\begin{aligned}
 dt(B) + \rho(B) \;  dB &= \frac{1}{2}  dq  \\
 dt(A) + \rho(A) \; dA &= -\frac{1}{2}  dq
\end{aligned}
\end{equation}

\begin{lem}The variation of $t(\alpha)$ satisfies the following integral equation
\begin{align*}
\big(I + \frac{1}{2\pi} K \big) * dt (\alpha) = \frac{i}{\pi} dH - \frac{1}{4 \pi} \big( \Theta(\alpha-B) + \Theta(\alpha-A) \big) \;  dq
\end{align*}
\end{lem}
\begin{proof}
Starting from equation (\ref{eq:star}) and using the previous calculation (\ref{eq:denInt}) with $f(\alpha) = \Theta(\alpha - \beta)$, we have
\begin{align*}
2 \pi \; dt( \alpha) =& 2 i \;  dH  -  dB \; \Theta(\alpha - B) \; \rho(B) +  dA \; \Theta(\alpha - A) \rho(A) \\ & - \Theta(\alpha - B) \; dt (B) + \Theta(\alpha - A) \; dt(A) - \int_C K(\alpha - \beta) \; dt(\beta) \; d \beta.
\end{align*}
Then substituting equations (\ref{deltatAB}) for $\rho(A)dA$ and $\rho(B)dB$ gives 
\begin{align*}
dt (\alpha) + \frac{1}{2\pi} K * dt (\alpha) =  \frac{i}{\pi}dH -  \frac{1 }{4 \pi} \Theta( \alpha - B) \; dq - \frac{1 }{4 \pi} \Theta(\alpha - A) \; dq
\end{align*} which can be rearranged to give the desired result.
\end{proof}

Define functions $ D_{\pm}(\alpha) $ as solutions to integral equations
\begin{equation}\label{eq:dpm}
\begin{aligned}
\left( I + \frac{1}{2\pi} K\right) * D_-(\alpha)& = \frac{1}{2\pi} \\
\left( I + \frac{1}{2\pi} K \right) * D_+(\alpha)& = \frac{1}{2 \pi} \Big( \Theta(\alpha - B) + \Theta(\alpha - A) \Big).
\end{aligned}
\end{equation}

\begin{lem} \label{lem:dressedcharge}The functions $ D_{\pm} $ satisfy
\begin{enumerate}[1)]
\item $ D_-(\alpha) = \frac{1}{2\pi} \left( 1 +  F(\alpha, B) - F(\alpha, A) \right) $
\item $ D_+(\alpha) = F(\alpha, B) + F(\alpha, A) $.
\end{enumerate}
\end{lem}
\begin{proof} ~ \\
\begin{enumerate}[1)]
\item By definition of $R$ we have
\begin{equation*}
\begin{aligned}
\left( I + \frac{1}{2\pi} K \right) * (I + R) = I,
\end{aligned} \; \; \; \; \begin{aligned} ( I + R) * \left( I + \frac{1}{2\pi} K \right) = I. \end{aligned}
\end{equation*}
Inverting $I+\frac{1}{2\pi}K$ in the definition of $D_-$ we obtain
\begin{align*}
D_-(\alpha) =& \big( (I + R) * \frac{1}{2 \pi} \big)(\alpha) \\ =& \frac{1}{2 \pi} + \frac{1}{2\pi}\int_A^B R(\alpha, \beta) \; d \beta \\
=& \frac{1}{2\pi} + \frac{1}{2 \pi}\int_A^B \frac{\partial F(\alpha, \beta) }{\partial \beta} \; d \beta
\\ =& \frac{1}{2\pi}\left(1+ F(\alpha, B) - F(\alpha, A)\right).
\end{align*}
\item Similarly for $D_+$ we obtain
\begin{align*}
D_+(\alpha) &= (I + R) * \frac{1}{2\pi} \left(\Theta(\alpha - B) + \Theta(\alpha-A)\right) \\
&= (I + R) * \left( \big(I + \frac{1}{2 \pi} K \big) * F(\alpha,B) + \big(I + \frac{1}{2\pi} K\big) * F(\alpha,A)\right) \\
&= F(\alpha, B) + F(\alpha, A).
\end{align*}
\end{enumerate}
\end{proof}

It is easy to see now that the variation of $t(\alpha)$ can be written as
\begin{lem} \label{lem:deltat}
\begin{equation}\label{deltat}
 dt (\alpha) = 2 i \; D_-(\alpha) \;  dH - \frac{1}{2}\; D_+(\alpha) \;  dq.
\end{equation}
\end{lem}

Note the following properties of the functions $D_\pm$:
\begin{enumerate}[a)]
\item When we vary $ q $ and $ H$, $ t(\alpha) $ remains real which implies that $  dt(\alpha) $ is real. It follows then that
\begin{align*}
\text{Re}\;  D_-(\alpha) = 0 \\
\text{Im} \; D_+(\alpha) = 0.
\end{align*}
\item The functions $D_\pm(\alpha)$ have the following symmetries with respect to the complex conjugation
\begin{align*}
D_-(\alpha)^* &= 1 + F(-\alpha^*, -\beta^*) + F(-\alpha^*, -A^*)  = D_-(-\alpha^*)  \\
D_+(\alpha)^* &= -F(-\alpha^*,A) - F(-\alpha^*, B) = -D_+(-\alpha^*).
\end{align*}
In particular,
\begin{equation}
\begin{aligned}
D_-(B) &= D_-(A) \\
D_+(B) &= -D_+(A).
\end{aligned}
\end{equation}
\end{enumerate}

\begin{lem} \label{lem:DalphaDer}
\begin{align*}
\frac{\partial D_+(\alpha)}{\partial \alpha} &= - R(\alpha, B)( 1 +  D_+(B)) -  R(\alpha, A ) (1-D_+(A) ) \\
\frac{\partial D_-(\alpha)}{\partial \alpha} &=  R(\alpha, A) D_-(A) - R(\alpha, B) D_-(B)
\end{align*}
\end{lem}
\begin{proof}
This follows from a simple computation using Lemmas \ref{lem:dressedcharge} and \ref{eq:pfpa}.
\end{proof}

Now, using the variation of $t(\alpha)$, let us complete the first variation of $g(q, H)$.

\begin{prop}\label{prop:firstderivatives} Partial derivatives of $g(q, H)$ can be written as
\begin{align*}
\frac{\partial g}{\partial H } &= - 2 i \int_C f'(\alpha) D_-(\alpha) \; d \alpha \\
\frac{\partial g}{\partial q} &= \frac{1}{2} \left( f(B) + f(A) \right) + \frac{1}{2} \int_C f'(\alpha) D_+(\alpha) \; d \alpha.
\end{align*}
\end{prop}
\begin{proof} Using equation (\ref{deltatAB}) and Lemma \ref{lem:deltat}, we have
\begin{align*}
dg(q,H) &= dB \; f(B) \; \rho(B) - dA \; f(A) \; \rho(A) + \int_C f(\alpha) \; \partial_\alpha\big( dt(\alpha) \big) \; d \alpha \\
&= \big( dB \; \rho(B) + dt (B) \big) f(B) - \big(dA \; \rho(A) + dt(A)\big) f(A) - \int_C f'(\alpha) \; dt(\alpha) \; d \alpha \\
&= \frac{1}{2} ( f(B) + f(A) ) \;  dq - \int_C f'(\alpha) \Big( 2 i \; D_-(\alpha) \; dH - \frac{1}{2}\; D_+(\alpha) \;  dq\Big) \; d\alpha \\
&= \frac{dq}{2} \Big( \left( f(B) + f(A) \right) + \int_C f'(\alpha) D_+(\alpha) d \alpha \Big) - 2 i\;   dH \int_C f'(\alpha) D_-(\alpha)  \; d \alpha
\end{align*}
from which the proposition follows.
\end{proof}

\subsection{Second derivatives of $ g(q, H ) $ }
Next we turn to calculating the second derivatives of $g(q,H)$. To begin we prove several lemmas that will be useful for later computations.
\begin{lem}\label{eq:delf}
\begin{align*}
dF(\alpha, \gamma) = R(\alpha, B)\; F(B, \gamma) \; dB - R(\alpha, A) \; F(A,\gamma) \; dA
\end{align*}
\end{lem}
\begin{proof}
Starting with the definition (\ref{eq:bigf}) and taking the total derivative we have
\begin{align*}
dF(\alpha, \gamma) + \frac{1}{2 \pi} & \int_A^B K(\alpha - \beta)  \;  dF(\beta,\gamma) \; d \beta  \\ & + \frac{1}{2 \pi} \Big( K(\alpha - B) F(B,\gamma) \;  dB - K(\alpha - A) F(A,\gamma) \;  dA \Big) = 0
\end{align*}
So
\begin{align}
\left(I + \frac{1}{2\pi} K  \right) *  dF(\alpha, \gamma) &= -\frac{1}{2\pi}\Big( K(\alpha - B) F(B, \gamma) \;  dB - K(\alpha - A) F(A,\gamma)  dA \Big)
\end{align}
Acting with $I + R $ on both sides and using equation (\ref{eq:rk}) gives the lemma.
\end{proof}

\begin{lem} \label{lem:DmDer}
For the total derivative of $D_-(\alpha)$ we have
\begin{align*}
dD_-(\alpha) = R(\alpha, B) \; D_-(B) \;  dB - R(\alpha, A) \; D_-(A) \;  dA. 
\end{align*}
\end{lem}
\begin{proof} From the formula for $D_-(\alpha)$in Lemma \ref{lem:dressedcharge} we get
\begin{align*}
 dD_-(\alpha) = \frac{1}{2\pi} \left(dF(\alpha, B) - dF(\alpha, A) + \frac{\partial F(\alpha, B)}{\partial B}  dB - \frac{\partial F(\alpha, A)}{\partial A}  dA\right).
\end{align*}
Substituting the result from Lemma \ref{eq:delf} into the variation of $D_-(\alpha)$ we have
\begin{align*}
dD_-(\alpha) =& \frac{1}{2\pi}\left(R(\alpha, B) \; F(B,B) \;  dB - R(\alpha, A) \; F(A,B) \;  dA\right. \\ & \;\; \left. - R(\alpha, B) \; F(B,A) \;  dB + R(\alpha, A) \; F(A,A) \;  dA 
+ \frac{\partial F(\alpha, B)}{\partial B} dB - \frac{\partial F(\alpha, A)}{\partial A} dA \right) \\
=& \frac{1}{2\pi}\left(R(\alpha, B) \big( F(B,B) - F(B,A) + 1\big) \; dB -
R(\alpha, A) \big( F(A,B) - F(A,A) + 1\big) \;  dA \right)\\
 =&  R(\alpha, B) \; D_-(B) \;  dB - R(\alpha, A) \; D_-(A) \;  dA 
\end{align*}
as desired.
\end{proof}

\begin{lem} \label{lem:DpDer} 
For the total derivative of $D_+(\alpha)$ we have
\begin{align*}
dD_+(\alpha) = R(\alpha, B) \; \big(1 + D_+(B)\big) \; dB + R(\alpha, A) \; \big(1 - D_+(A)\big) \;  dA
\end{align*}
\end{lem}
\begin{proof} By straightforward computation starting from Lemma \ref{lem:dressedcharge} we obtain
\begin{align*}
 dD_+(\alpha) =&  dF(\alpha, B) +  dF(\alpha, A) +  dB \; \frac{\partial F(\alpha, B)}{\partial B} +  dA \; \frac{\partial F(\alpha, A)}{\partial A}  \\
=& dF(\alpha, B) +  dF(\alpha, A) +  dB \; R(\alpha, B) +  dA \; R(\alpha, A) \\
=&  \;  R(\alpha, B) \; dB +  R(\alpha, A) dA + R(\alpha, B) F(B,B) \; dB - R(\alpha, A) F(A,B) \; dA \\
 & \; \; +R(\alpha, B) F(B,A) \; dB - R(\alpha, A) F(A,B) \; dA \\
 =& \; R(\alpha, B) \big( 1 + F(B,B) + F(B,A) \big) \; dB  + R(\alpha, A) \big( 1 - F(A,B) - F(A,A) \big) \; dA.
\end{align*}
Using Lemma \ref{lem:dressedcharge} again gives the desired result.
\end{proof}

\begin{lem}  \label{lem:J}
\begin{align*}
 ( 1 + D_+(B)) \; D_-(B) = \frac{1}{2 \pi}
\end{align*}

\end{lem}
\begin{proof}
Let $J=( 1 + D_+(B)) \; D_-(B)$. It could be that $D_+,D_-$ depend on $A=-a+ib, B=a+ib$. Note, however, that $F(\alpha,\gamma)$ is invariant under a pure imaginary shift in $\alpha,\gamma,$ and the contour, so $D_+(B),D_-(B)$ can only depend on $a$. When $a=0$, we have $F(\alpha,\gamma) = \frac{1}{2\pi} \Theta(\alpha-\gamma)$ and it follows $J=\frac{1}{2\pi}$. It is left to show that the full derivative of $J$ with respect to $a$ is zero.

Taking derivatives we have
\begin{align*}
&\dd{}{a}D_-(B) = \frac{\partial D_-(\alpha)}{\partial \alpha}\Big|_{\alpha = B}\frac{\partial B}{\partial a} + \frac{\partial D_-(\alpha)}{\partial A}\Big|_{\alpha = B}  \frac{\partial A}{\partial a} + \frac{\partial D_-(\alpha)}{\partial B}\Big|_{\alpha = B}  \frac{\partial B}{\partial a} \\
&\dd{}{a}D_+(B) = \frac{\partial D_+(\alpha)}{\partial \alpha}\Big|_{\alpha = B}\frac{\partial B}{\partial a} + \frac{\partial D_+(\alpha)}{\partial A}\Big|_{\alpha = B}  \frac{\partial A}{\partial a} + \frac{\partial D_+(\alpha)}{\partial B}\Big|_{\alpha = B}  \frac{\partial B}{\partial a}.
\end{align*}
Using Lemmas \ref{lem:DalphaDer}, \ref{lem:DmDer}, and \ref{lem:DpDer} for the derivatives and recalling $D_\pm(A) = \mp D_\pm(B) = \mp D_\pm$, these become
\begin{equation*}
\begin{aligned}
    \dd{}{a}D_-(B)   =& R(B,A) D_-(A) - R(B,B) D_-(B) + R(B,A)D_-(A) + R(B,B) D_-(B)\\
      =& 2 R(B,A) D_-(B)\\
    \dd{}{a}D_+(B)  =& -R(B,B) (1-D_+(A)) - R(B,A) (1+D_+(B)) \\ & \;\; - R(B,A) (1-D_+(A)) + R(B,B)(1+D_+(B)) \\
      =& -2 R(B,A) (1+D_+(B)).
\end{aligned}
\end{equation*}
From which it follows
\begin{align*}
\dd{}{a}J & = D_-(B)\dd{}{a}D_+(B) + (1+D_+(B))\dd{}{a}D_-(B)  \\ 
& = -2 D_-(B) R(B,A) (1+D_+(B)) + 2 (1+D_+(B)) R(B,A) D_-(B) \\
& = 0.
\end{align*}
We see that $(1+D_+(B))D_-(B) = \frac{1}{2\pi}$, independent of the endpoints of the contour. 
\end{proof}

Introduce
\begin{equation}\label{eq:xi}
\begin{aligned}
\xi = 2 \pi i \Big( f'(B) + \int_C f'(\alpha) R(\alpha, B)  \; d \alpha \Big),  \\
\widetilde{\xi} = 2 \pi i \Big( f'(A) + \int_C f'(\alpha) R(\alpha, A) \;  d \alpha \Big).
\end{aligned}
\end{equation}
Note that since $ f(\alpha)^* = f(-\alpha^*) $ we have $ \widetilde{\xi} = \xi^* $.

\begin{prop} The following hold
\begin{align} \label{eq:eq1}
d\frac{\partial g}{\partial H}  &= - \frac{1}{\pi} D_-(B) \; \xi \; dB + \frac{1}{\pi} \; D_-(A)\;  \widetilde{\xi} \; dA  \\
\label{eq:eq2}
 d\frac{\partial g}{\partial q}  &= - \frac{i}{4 \pi} \big( 1 + D_+(B) \big) \xi \;  dB - \frac{i}{4 \pi } \big( 1 - D_+(A) \big) \w \xi \;  dA.
\end{align}
\end{prop}
\begin{proof}
We first prove (\ref{eq:eq1}). Starting from Prop. \ref{prop:firstderivatives} we have
\begin{align*}
d\frac{\partial g}{\partial H} =& - 2 i\; d\int_C f'(\alpha) D_-(\alpha) d \alpha \\
=& - 2 i \left(  dB  \; f'(B) \; D_-(B) -  dA \; f'(A)\; D_-(A) + \int_C f'(\alpha)  dD_-(\alpha) \; d \alpha \right)\\
=& - 2 i \Big(  dB \; f'(B)\;  D_-(B) -  dA \; f'(A) D_-(A)  \\ & \; \; \; \; + \int_C f'(\alpha) \; \big( R(\alpha, B) \; D_-(B) \; dB - R(\alpha, A)\;  D_-(A) \;  dA \big) \; d \alpha \Big) \\
 =& - 2 i \Big(  dB \; D_-(B) \big( f'(B) + \int_C f'(\alpha) R(\alpha, B) d \alpha\big)  \\ & \; \; \; \; \; \; -  dA \; D_-(A) \big( f'(A) + \int_C f'(\alpha) R(\alpha, A) \; d \alpha \big)  \Big)
\end{align*}
 where we use Lemma \ref{lem:DmDer}. Now let us prove equation (\ref{eq:eq2}). From Prop. \ref{prop:firstderivatives} we have
\begin{align*}
 d\frac{\partial g}{\partial q} =& \frac{1}{2} ( f'(B) \; dB + f'(A) \; dA ) + \frac{1}{2} ( f'(B) D_+(B) \;  dB - f'(A) D_+(A) \; dA )  \\ & + \frac{1}{2} \int_C f'(\alpha) \Big( R(\alpha,B)\;  (1 + D_+(B) ) \;  dB + R(\alpha, A) \;  ( 1 - D_-(A) ) \; dA \Big) d \alpha  \\
=& \frac{1}{2}  \Big( f'(B)  \big(1 + D_+(B) \big) + (1 + D_+(B) ) \int_C f'(\alpha) R(\alpha, B) d \alpha \Big) \;  dB \\ & + \frac{1}{2} \Big( f'(A) \big(1 - D_-(A) \big) + (1 - D_-(A) ) \int_C f'(\alpha) R(\alpha, A) d \alpha  \Big) \;  dA
\end{align*}
where we use Lemma \ref{lem:DpDer}.
\end{proof}

Now let us express $ dA ,  dB $ in terms of $ dH,  dq $.
Substituting the formula (\ref{deltat}) for $ dt(\alpha) $ into equation (\ref{deltatAB}) we obtain
\begin{align*}
2 i \; D_-(B) \;  dH - \frac{1}{2}\; D_+(B) \; dq + \rho(B) dB &= \frac{1}{2} dq  \\
2 i \; D_-(A) \;  dH - \frac{1}{2}\; D_+(A) \; dq + \rho(A) dA &= -\frac{1}{2} dq
\end{align*}
or
\begin{equation}\label{eq:dab}
\begin{aligned}
dB &= \left( \frac{- 2 i D_-(B) }{\rho(B) }\right) dH + \left( \frac{1+ D_+(B) }{2\rho(B)}\right) dq \\
dA &= \left( \frac{-2 i D_-(A) }{\rho(A)} \right) dH -  \left(\frac{1- D_+(A)}{2\rho(A)} \right) dq.
\end{aligned}
\end{equation}
\begin{prop} \label{prop:gSD1}
The second derivatives of $g(q,H)$ are given by
\begin{align*}
\frac{\partial^2 g}{\partial H^2} &= \frac{2 i}{\pi}  \left( \frac{D_-(B)^2 }{\rho(B) } \xi  - \frac{D_-(A)^2 }{\rho(A) } \w \xi  \right) \\
\frac{\partial^2 g}{\partial H \partial q} &= - \frac{1}{2 \pi} \left( \frac{1+ D_+(B) }{\rho(B)}\;  D_-(B) \; \xi  + \frac{1- D_+(A) }{\rho(A)}\;  D_-(A)  \; \w \xi  \right) \\
\frac{\partial^2 g}{\partial q^2} &= - \frac{i}{8 \pi} \left( \frac{(1  + D_+(B))^2}{\rho(B)} \; \xi - \frac{(1 - D_+(A))^2 }{\rho(A)} \; \w \xi \right) \\
\end{align*}
\end{prop}
\begin{proof}
Combining (\ref{eq:dab}) with formulas for the second variation, (\ref{eq:eq1}) and (\ref{eq:eq2}), we obtain
\begin{align*}
d\frac{\partial g}{\partial H}  = &  - \frac{1}{\pi} \left(\left( \frac{- 2 i D_-(B) }{\rho(B) }\right) d H +  \frac{1+ D_+(B) }{2\rho(B)} d q \right) \; D_-(B) \; \xi\\&  + \frac{1}{\pi} \left(  \left( \frac{-2 i D_-(A) }{\rho(A)} \right) d H -  \frac{1- D_+(A)}{2\rho(A)} d q \right) \; D_-(A)\;  \widetilde{\xi} \\
= & \; i \left( \frac{D_-(B)^2 }{\rho(B) } \xi  - \frac{D_-(A)^2 }{\rho(A) } \w \xi  \right) d H \\
& - \frac{1}{4} \left( \frac{1+ D_+(B) }{\rho(B)}\;  D_-(B) \; \xi  + \frac{1- D_+(A) }{\rho(B)}\;  D_-(A) \; \w \xi  \right) d q
\end{align*}
and
\begin{align*}
d \frac{\partial g}{\partial q}  = & - \frac{i}{4 \pi} \big( 1 + D_+(B) \big) \; \xi \; \left( \left( \frac{- 2 i D_-(B) }{\rho(B) }\right) d H +  \frac{1+ D_+(B) }{2\rho(B)} d q \right) \\ & - \frac{i}{4 \pi} \big( 1 - D_+(A) \big) \; \w \xi \; \left( \left( \frac{-2 i D_-(A) }{\rho(A)} \right) d H -  \frac{1- D_+(A)}{2\rho(A)} d q \right) \\
= & - \frac{1}{2 \pi} \left( ( 1 + D_+(B)) \; \frac{D_-(B)}{\rho(B)} \; \xi +  ( 1 - D_+(A)) \; \frac{D_-(A)}{\rho(A)}\; \w \xi \right) d H \\
& - \frac{i}{8 \pi} \left( \frac{(1  + D_+(B))^2}{\rho(B)} \; \xi- \frac{(1 - D_+(A))^2 }{\rho(A)} \; \w \xi  \right) \; d q.
\end{align*}
\end{proof}

Recall that the quantity we are interested in is equation (\ref{eq:HqH}) in which the role of $f(\alpha)$ is played by $\psi_u^{\pm}(\alpha)$. The functions $\psi_u^{\pm}(\alpha)$ depend on the spectral parameter $u$ as $\psi^{\pm}(\alpha + i u)$. We will assume then that $ f $ depends on the spectral parameter in the same way. In particular, $ \frac{\partial f}{\partial u} = i f'$.
\begin{prop} \label{prop:gSD2}
\begin{align*}
\frac{\partial^2 g}{\partial H \partial u }  &= - \frac{i}{\pi} \left( D_-(B) \xi  - D_-(A) \w \xi \right) \\
\frac{\partial^2 g}{\partial q \partial u }  &= \frac{1}{4 \pi} \left(  (1 + D_+(B) ) \; \xi  +   (1 - D_+(A) ) \; \w \xi \right)
\end{align*}
\end{prop}
\begin{proof}
Starting with Proposition \ref{prop:firstderivatives}
\begin{align*}
\frac{\partial g}{\partial H } &= - 2 i \int_C f'(\alpha) D_-(\alpha) \; d \alpha \\
&=-2 i \Big( f(B) D_-(B) - f(A) D_-(A) - \int_C f(\alpha) D'_-(\alpha) \; d \alpha \Big)\\
&= -2 i \Big( f(B) D_-(B) - f(A) D_-(A) + \int_C f(\alpha) \big( R(\alpha, B) D_-(B)  - R(\alpha,A ) D_-(A) \big) \; d \alpha \Big).
\end{align*}
Differentiating with respect to $ u $ gives
\begin{align*}
\frac{\partial^2 g}{\partial H \partial u} &= 2 \Big( f'(B) D_-(B) - f'(A) D_-(A) + \int_C f'(\alpha) \big( R(\alpha, B) D_-(B)  - R(\alpha, A) D_-(A) \big) \; d \alpha \Big) \\
&= 2 D_-(B)  \Big( f'(B) + \int_C f'(\alpha) R(\alpha, B) \; d \alpha \Big) - 2 D_-(A) \Big( f'(A) + \int_C f'(\alpha) R(\alpha, A)  \; d \alpha \Big) \\
&= - \frac{ i}{\pi} \left( D_-(B) \; \xi  - D_-(A) \; \w \xi \right).
\end{align*}

Similarly
\begin{align*}
\frac{\partial g}{\partial q} =& \; \frac{1}{2} \left( f(B) + f(A) \right) + \frac{1}{2} \int_C f'(\alpha) D_+(\alpha) \; d \alpha \\
=& \; \frac{1}{2} f(B) ( 1 + D_+(B) ) + \frac{1}{2} f(A) ( 1 - D_-(A) ) - \frac{1}{2} \int_C f(\alpha) D_+'(\alpha) \; d \alpha \\
=& \; \frac{1}{2} f(B) ( 1 + D_+(B) ) + \frac{1}{2} f(A) ( 1 - D_-(A) ) + \frac{1}{2} \int_C f(\alpha) R(\alpha,B) ( 1 + D_+(B) ) d \alpha  \\ & \; + \frac{1}{2} \int_C f(\alpha) R(\alpha, A) ( 1 - D_+(A) ) d \alpha.
\end{align*}
Differentiating with respect to $ u $ gives
\begin{align*}
\frac{\partial^2 g}{\partial q \partial u} = \frac{1}{4 \pi} \;\left( (1 + D_+(B) ) \; \xi  + \; (1 - D_+(A) ) \; \w \xi \right).
\end{align*}
\end{proof}

\section{Limit Shape Shape and its integrability}\label{sec:4}

\subsection{The variational principle}\label{sec:varprinc}
Using the same arguments as in the case of dimer models \cite{CKP}, and in the homogeneous six vertex model \cite{ZJ,PR,RS}, one can formulate the variational principle for the six vertex model with inhomogeneous weights.

Introduce the height function for the six vertex model in the usual way, see for example \cite{RS}. It is an integer valued function $\theta_{n,m}$ defined on faces of the square grid. Its value on a face is determined by local rules shown in Fig. \ref{fig:6Vheight}. After choosing the value $\theta_{a,b}$ at a single face, the local rules give a bijection between possible height functions and six vertex configurations on  a planar simply connected domain.

However, local rules do not define a height function on a cylinder. In this case, the local definition of the height function may lead to a discontinuity along a noncontractible  path. This simply means that the global object defined by these local rules is not a function but a section of the corresponding line bundle.

To have a height function be a true function we cut the cylinder at column $n=0$ and fix its value at a face with the condition $\theta_{0,0}=0$. Thus, the height function for us is a function $\theta_{n,m}$ with $n=0,1,\dots, N, m=0,1,\dots, M$. If a six vertex configuration has $n_0$ paths entering from the bottom and exiting from the top, the height function has monodromy
\[
\theta_{N,m} -\theta_{0,m} =n_0.
\]

\begin{figure}[h]
\begin{tabular}{cccccc}
\begin{tikzpicture}
\draw[ultra thick] (-1,0) -- (0,0); \draw[ultra thick] (0,0) -- (1,0);
\draw[ultra thick] (0,-1) -- (0,0); \draw[ultra thick] (0,0) -- (0,1);
\node[font=\small] at (-0.5,-0.5) {$\theta$};
\node[font=\small] at (0.5,-0.5) {$\theta + 1$};
\node[font=\small] at (-0.5,0.5) {$\theta + 1$};
\node[font=\small] at (0.5,0.5) {$\theta + 2$};
\end{tikzpicture}
&
\begin{tikzpicture}
\draw (-1,0) -- (0,0); \draw (0,0) -- (1,0);
\draw (0,-1) -- (0,0); \draw (0,0) -- (0,1);
\node[font=\small] at (-0.5,-0.5) {$\theta$};
\node[font=\small] at (0.5,-0.5) {$\theta$};
\node[font=\small] at (-0.5,0.5) {$\theta$};
\node[font=\small] at (0.5,0.5) {$\theta$};
\end{tikzpicture} 
&
\begin{tikzpicture}
\draw[ultra thick] (-1,0) -- (0,0); \draw[ultra thick] (0,0) -- (1,0);
\draw (0,-1) -- (0,0); \draw (0,0) -- (0,1);
\node[font=\small] at (-0.5,-0.5) {$\theta$};
\node[font=\small] at (0.5,-0.5) {$\theta$};
\node[font=\small] at (-0.5,0.5) {$\theta + 1$};
\node[font=\small] at (0.5,0.5) {$\theta + 1$};
\end{tikzpicture}
&
\begin{tikzpicture}
\draw (-1,0) -- (0,0); \draw (0,0) -- (1,0);
\draw[ultra thick] (0,-1) -- (0,0); \draw[ultra thick] (0,0) -- (0,1);
\node[font=\small] at (-0.5,-0.5) {$\theta$};
\node[font=\small] at (0.5,-0.5) {$\theta + 1$};
\node[font=\small] at (-0.5,0.5) {$\theta$};
\node[font=\small] at (0.5,0.5) {$\theta + 1$};
\end{tikzpicture}
&
\begin{tikzpicture}
\draw[ultra thick] (-1,0) -- (0,0); \draw (0,0) -- (1,0);
\draw (0,-1) -- (0,0); \draw[ultra thick] (0,0) -- (0,1);
\node[font=\small] at (-0.5,-0.5) {$\theta$};
\node[font=\small] at (0.5,-0.5) {$\theta$};
\node[font=\small] at (-0.5,0.5) {$\theta + 1$};
\node[font=\small] at (0.5,0.5) {$\theta $};
\end{tikzpicture}
&
\begin{tikzpicture}
\draw (-1,0) -- (0,0); \draw[ultra thick] (0,0) -- (1,0);
\draw[ultra thick] (0,-1) -- (0,0); \draw (0,0) -- (0,1);
\node[font=\small] at (-0.5,-0.5) {$\theta$};
\node[font=\small] at (0.5,-0.5) {$\theta + 1$};
\node[font=\small] at (-0.5,0.5) {$\theta$};
\node[font=\small] at (0.5,0.5) {$\theta $};
\end{tikzpicture}

\end{tabular}
\caption{Local rules for the height function of the six vertex model.}
\label{fig:6Vheight}
\end{figure}
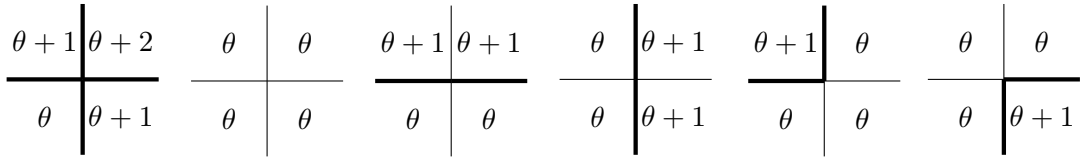

Let us describe the thermodynamic limit for a torus and the formation of the limit shape. 

Fix a sequence $\{\epsilon_k\}_{k=1}^\infty$ such that $\epsilon_k\to 0$ as $k\to \infty$ and $\epsilon_k>0$. One should think of it as a sequence of mesh lengths. Consider a sequence of cylinders of size $N_k\times M_k$ such that 
$N_k\epsilon_k\to L$ and $M_k\epsilon_k\to T$ as $k\to \infty$. Fix two functions $\phi_1,\phi_2: [0,L]\to \R$ satisfying conditions
\[
|\phi_i(x)-\phi_i(x')|\leq |x-x'|
\]
assuming that $\phi_i(L)=\phi_i(0)+q$ for some $0<q<1$.

Fix a sequence of nonnegative  integers $n_0^k$ such that $\epsilon_kn^k_0\to q$ as $k\to \infty$.
For each $k$ fix functions $\theta^{k,1},\theta^{k,2}: \{0, 1,\dots, N_k\}\to \ZZ $ such that
\[
|\theta^{k,i}_{n}-\theta^{k,i}_{n,}|\leq |n-n'|, \ \ \theta^{k,i}_{N_k}-\theta^{k,i}_0=n^k_0
\]
Assume that as $k\to \infty$, normalized boundary height functions $\epsilon_k\theta^{k,i}$, regarded as a piece-wise linear function on $[0,L]$, converges to $\phi_i$ in $L_\infty$ topology. Also, assume that $u_i=u(\epsilon_k i)$ and $v_i=v(\epsilon_k i)$
where $u(x)$ and $v(y)$ are smooth functions. 

Define the space $H_{L,T,q}^{\phi_1,\phi_2}$ of asymptotic height functions as the space of mappings
$h:[0,L]\times [0,T]\to \R$ with the properties
\[
|h(x,y) - h(x',y)| <  |x-x'|, \ \ |h(x,y) - h(x,y')| <  |y-y'|
\]
\[
h(L,y)=h(0,y)+q, \ \ h(x,0)=\phi_1(x), \ \ h(x,T)=\phi_2(x)
\]
where $0\leq q\leq 1$.

Applying the same arguments as \cite{CKP} we expect that the sequence of random variables $\theta_{n,m}^k$ (height functions for the six vertex  model on $N_k\times M_k$ cylinders with boundary conditions $\theta^{k,1}$ and $\theta^{k,2}$) have the following asymptotic:
\[
\epsilon_k\theta_{n,m}\to h_0(x,y)+\epsilon_k \phi(x,y)+\dots
\]
Here  we assume that $\epsilon_kn\to x, \epsilon_km\to y$.
The deterministic functions $h_0(x,y)$ is the limit shape and the random variable $\phi(x,y)$  describes Gaussian fluctuations around the limit shape. The convergence is in probability, 
with respect to to the six vertex probability measures on 
$N_k\times M_k$ cylinders.

The limit shape $h_0$ in the inhomogeneous six vertex model is the minimizer of the functional
\begin{align} \label{eq:st}
S[h] = \int_0^L \int_0^T \; \sigma_{u(y) - v(x)} (\partial_x h, \partial_y h) dydx
\end{align}
in the space $H_{L,T,q}^{\phi_1,\phi_2}$.

\subsection{Hamiltonian formulation}
As was done in the homogeneous case \cite{RS} let us reformulate the variational principle in the Hamiltonian framework. For the time being we will assume that the six vertex model is homogeneous in the vertical direction, i.e. $u(y)=u$.

Define the space of asymptotic height functions $\mathcal{H}_{L,q}$ in one horizontal layer as the space of mappings $h: [0,L]\to \R$ satisfying the Lipschitz and periodicity conditions
\[
|h(x) - h(x')<|x-x'|, \ \ h(L,y) = h(0,y) + q,
\]
where $0<q<1$ is fixed. Elements of the cotangent space $T^*\mathcal{H}_{L,q}$ can be identified with pairs of functions $(\pi(x),h(x))$ where $h(x)$ is as above and $\pi(x)$ is a function periodic in $x$. It is an infinite dimensional symplectic manifold with symplectic form
\begin{equation}
\omega = \int_0^L \delta \pi(x) \wedge \delta h(x) dx.
\end{equation}

Now let us formulate the Hamiltonian version of the variational principal. It is an easy exercise to check that if $(\pi(x,y),h(x,y))$ is a flow line of the Hamiltonian vector field on $T^* \mathcal{H}_{L,q}$ generated by the Hamiltonian function
\begin{equation} \label{main-Ham}
H_u(\pi,h) = \int_0^L \mathcal{H}_{u-v(x)}(\pa_xh(x), \pi(x))dx.
\end{equation}
then $h(x,y)$ is a solution to the Euler-Lagrange equations. Here the function $\mathcal{H}_u(q,H)$ is the semigrand canonical free energy (\ref{eq:semigrand}).

Moreover, there is unique flow line $(\pi_0,h_0)$ of this Hamiltonian vector field  
which connects Lagrangian subspaces $T^*_{\phi_1}\mathcal{H}_{q,L}$ and $T^*_{\phi_2}\mathcal{H}_{q,L}$ in Euclidian  time $T$. And the component $h_0$ is the minimizer of \ref{eq:st}. 
Also, the pair $(\pi_0,h_0)$ is the unique critical point in $T^* H_{L,T,q}^{\phi_1,\phi_2}$ of the Hamilton-Jacobi functional
\begin{equation} \label{eq:HJ}
S_{HJ} (\pi,h) = \int_0^L\int_0^T (\pi(x,y)\partial_y h(x,y) - \mathcal{H}_{u-v(x)}(\pa_xh(x,y), \pi(x,y)))dydx
\end{equation}
Indeed, first minimizing in $\pi$ and evaluating this functional at the unique critical point we will
arrive to the minimization problem for (\ref{eq:st}). The existence and uniqueness of the
critical point follows from the convexity of $\mathcal{H}_u(s,t)$ in $t$.

\subsection{Poisson commutativity of the Hamiltonians}
The main result of this paper is the Poisson commutativity of the family $H_u(\pi,h)$ where $v(x)$ is an arbitrary smooth functions as described in \ref{sec:varprinc}. The first step is the following proposition that characterizes the commutativity of the Hamiltonians $ H_u $ in terms of differential identities for the densities $ \H_{u-v(x)}$.

\begin{thm}\label{main1} The Hamiltonians \ref{main-Ham} form Poisson commutative family
\[
\{ H_u, H_w \} = 0
\]
for any $u$ and $w$ if the following identities hold for $\H=\H(\pi, h, u-v(x))$ and $\w \H=\H(\pi,h,w-v(x))$:
\begin{equation} \label{eq:commH}
\begin{aligned}
& \H_{11} \wH_{22} - \wH_{11} \H_{22} = 0 \\
& \H_{11} \w \H_{23} + \H_{12} \w \H_{13} - \H_{13} \w \H_{12} - \H_{23} \w \H_{11} = 0 \\
& \H_{12} \w \H_{23} + \H_{22} \w \H_{13} - \H_{13} \w \H_{22} - \H_{23} \w \H_{12} = 0
\end{aligned}
\end{equation}
where $f_i$ means the derivative of $f$ with respect to the $i$-th argument.
\end{thm}
\begin{remark}
The assumption that $ v(x) $ is smooth can be relaxed; for example, the computations are almost identical when $ v(x) $ is piecewise smooth, but with some subtleties that we will not discuss here.
\end{remark}
\begin{proof}

Straightforward computation gives
\begin{align} \label{eq:bracket}
 \{ H_u, H_w \} = \int_0^L \left( A \; \partial_x \pi + B \; \partial_x^2 h + C \; \partial_x (u-v(x))\right) \; dy
\end{align}
where
\begin{equation}
\begin{aligned}
A = \H_{11} \wH_2 - \wH_{11} \H_2  \\
B = \H_{12} \wH_2 - \wH_{12} \H_2 \\
C = \H_{13} \wH_2 - \wH_{13} \H_2.
\end{aligned}
\end{equation}
The integrand in (\ref{eq:bracket}) will be a total derivative if all the mixed derivatives are zero (a closed one form), that is
\begin{equation}
\begin{aligned}
\partial_2 A - \partial_1 B &= 0 \\
\partial_3 A - \partial_1 C &= 0 \\
\partial_3 B - \partial_2 C &= 0.
\end{aligned}
\end{equation}

\end{proof}

In section \ref{Thmproof} , we will prove that the identities (\ref{eq:commH}) hold for the 6-vertex model.

\subsection{Proof of Commuting Hamiltonians} \label{Thmproof}
Let us apply the analysis of integral equations in the previous section to the case of interest eqn. (\ref{eq:HqH}), reproduced here
\begin{align*} 
\mathcal{H}_{u-v(x)}^\pm(q, H) = \pm H + l_{\pm} + \int_{C} \psi_u^{\pm}(\alpha) \rho(\alpha) d\alpha
\end{align*}
where $l_\pm $ depend only on the spectral parameter. Note that when we take second derivatives, only the contributions from the integral above will remain. This integral is precisely of the form of $g(q,H)$. Using Prop. \ref{prop:gSD1} and \ref{prop:gSD2} and Lemma \ref{lem:J}, we have
\begin{alignat*}{2}
\mathcal{H}_{11} &= -\frac{4 D_-^2}{\pi} \; \; &&\text{Im} ( \xi / q ) \\
\mathcal{H}_{12} &= - 2 && \text{Re} (\xi / q )\\
\mathcal{H}_{22} &= \frac{ \pi }{D_-^2} &&  \text{Im} (\xi/ q ) \\
\mathcal{H}_{13} &= \frac{2 D_- }{\pi} && \text{Im} ( \xi )  \\
\mathcal{H}_{23} &= \frac{1}{D_-} && \text{Re}( \xi )
\end{alignat*} 
where we write $D_\pm$ for $D_\pm(B)$. As these identities hold for both $\mathcal{H}^+$ and $\mathcal{H}^-$, they hold the max as well. Note that all these quantities depend on the spectral parameter.

\begin{thm} \label{main} The identities (\ref{eq:commH}) hold for the 6-vertex model.
\end{thm}
\begin{proof}
Start with the second equation of (\ref{eq:commH}). Substituting the above expressions for $ \mathcal{H} $ give
\begin{align*}
& \H_{11} \w \H_{23} + \H_{12} \w \H_{13} - \H_{13} \w \H_{12} - \H_{23} \w \H_{11}  \\
=& -\frac{4 D_-}{\pi} \left( \text{Im} (\xi / q ) \; \text{Re}(  \w \xi ) + \text{Re} (\xi / q ) \; \text{Im} (\w \xi ) - (\;  \w\cdot \leftrightarrow \cdot \; ) \right) \\
=& \frac{4 D_-}{\pi} \left( \text{Im} (\xi  \w \xi / q )  -   \text{Im} (\w \xi  \xi / q )    \right) = 0.
\end{align*}
Here we used the fact that $q$ and $H$ are fixed (the same for $\H$ and $\w \H$), and that the roots of the Bethe equations do not depend on the spectral parameter $ u $ (so that $D_-$ does not depend on $u$).

Similarly, for the third equation
\begin{align*}
& \H_{12} \w \H_{23} + \H_{22} \w \H_{13} - \H_{13} \w \H_{22} - \H_{23} \w \H_{12} \\
&= -\frac{2}{D_-} \left( \text{Re} (\xi / q ) \; \text{Re}( \w \xi ) - \text{Im} (\xi/ q ) \; \text{Im} (\w \xi ) - ( \; \w\cdot \leftrightarrow \cdot \; ) \right) \\
&= -\frac{2}{D_-} \left( \text{Re} (\xi \w \xi / q )  -  \text{Re} (\xi \w \xi / q )  \right) = 0.
\end{align*}

The first equation follows from commutativity of homogeneous Hamiltonians. Reproducing this for completeness, we have
\begin{align*}
    &\H_{11} \w \H_{22} - \w \H_{11} \H_{22} \\
    & = -\frac{4D_-^2}{\pi} \text{Im} ( \xi / q) \frac{\pi}{D_-^2} \text{Im} (\w \xi / q) +  -\frac{4D_-^2}{\pi} \text{Im} (\w \xi / q) \frac{\pi}{D_-^2} \text{Im} (\xi / q) \\
    & = -4 \text{Im} ( \xi / q)\text{Im} (\w \xi / q) + 4 \text{Im} (\w \xi / q) \text{Im} (\xi / q) \\
    & = 0.
\end{align*}
\end{proof}

\section{Concluding remarks}\label{concl}

We proved the Poisson commutativity of the family of Hamiltonians (\ref{main-Ham}). It implies that the Euler-Lagrange equations, regarded as an evolution equation in Euclidean time, have infinitely many conservation laws. Thus, we can conjecture that the equations describing the 6-vertex model with integrable inhomogeneities in the horizontal direction, and which is homogeneous in the vertical direction, is integrable. 

When the model is inhomogeneous in the vertical direction, i.e. when $u(y)$ is not a constant function of $y$, we have the same conservation laws for the time dependent Hamiltonian
\[
H_{u(y)}(\pi,h) = \int_0^L \mathcal{H}_{u(y)-v(x)}(\pa_xh(x), \pi(x))dx.
\]
This fact has a simple illustration in the finite dimensional case. 

Assume $(M,\omega)$ is a symplectic manifold and and $H_1,\dots, H_n$ are Poisson commuting functions on $M$. Consider a time dependent Hamiltonian 
\[
H_t(x)=f(H_1(x), \dots, H_n(x), t)
\]
where $f$ is a smooth function. It is easy to check that functions $H_1, \dots, H_n$ remain constant along Hamiltonian flow lines generated by $H_t$, i.e. along solutions to differential equations
\[
\frac{dx^i(t)}{dt}=\sum_j (\omega^{-1}(x(t))^{ij}\frac{\pa H_t }{\pa x^i} (x(t)
\]
where this is written in local coordinates $x^i$ on $M$ and $\omega^{-1}(x)$ is the inverse matrix to the tensor $\omega_{ij}(x)$ representing  $\omega(x)=\sum_{ij}\omega_{ij}(x)dx^i\wedge dx^j$, $\sum_j \omega^{-1}(x)^{ij}\omega_{jk}(x)=\delta^i_k$.

We were focusing on the case $\Delta<-1$ and small  inhomogeneities. Extending these results to $-1\leq \Delta\leq 1$ and to $\Delta>1$ is straightforward. In the describing the Hamiltonian framework we were assuming smoothness of $\mathcal{H}_u(q,H)$ in $q$ and $H$ which is not alway true. However, this assumption is valid when the values of $(\pa_xh,\pa_yh)$ are in the disordered region. This can always be achieved by choosing appropriate initial  and target values of the height function. 

The very interesting  problem of finding out whether the Hamiltonian flow in question is really integrable remains. We conjecture that this is the case. To prove this one should find corresponding action angle variables as in other Hamiltonian systems \cite{FT1}.

\end{document}